\documentclass[twoside,12pt,reqno]{amsart}
\usepackage{amsmath,amsthm, amsfonts,amssymb}
\usepackage{color,soul,enumerate}
\usepackage[dvipsnames]{xcolor}  
\usepackage[all]{xy}
\usepackage{graphicx}
\usepackage{indentfirst}
\usepackage{bm}
\usepackage{mathrsfs}
\usepackage{latexsym}
\usepackage{hyperref}
\usepackage{microtype}	
\usepackage{bookmark}
\usepackage{tikz}[2010/10/13]
\hypersetup{
    colorlinks = true,
    linkcolor = blue,
    urlcolor = blue,
    citecolor = blue,
    anchorcolor = black,
    pdftoolbar = true
    linkbordercolor = {white},
}

\newcommand{\FS}{\mathfrak{F}_{\text{{s}}}}

\newcommand{\C}{\mathscr{C}}
\newcommand{\Irr}{Irr}
\newcommand{\ONB}{ONB}
\newcommand{\Hil}{\mathcal{H}}
\newcommand{\Pla}{\mathfrak{P}}

\newcommand{\be}{\begin{equation}}
\newcommand{\ee}{\end{equation}}
\newcommand{\beq}{\begin{eqnarray}}
\newcommand{\eeq}{\end{eqnarray}}
\newcommand{\beqs}{\begin{eqnarray*}}
\newcommand{\eeqs}{\end{eqnarray*}}
\newcommand{\abs}[1]{\left| #1 \right|}

\theoremstyle{plain}
\newtheorem{theorem}{Theorem}[section]

\newtheorem{lemma}[theorem]{Lemma~}

\newtheorem{proposition}[theorem]{Proposition~}
\newtheorem{corollary}[theorem]{Corollary~}
\newtheorem{definition}[theorem]{Definition~}

\newtheorem*{remark}{Remark}

\renewcommand{\geq}{\geqslant}
\renewcommand{\leq}{\leqslant}

\newcommand{\lra}[1]{\left\langle#1\right\rangle}

\title{Reflection positivity and Levin-Wen models}
\author{Arthur Jaffe}
\address{Harvard University, Cambridge, MA 02138}
\email{jaffe@g.harvard.edu}
\author{Zhengwei Liu}
\address{Harvard University, Cambridge, MA 02138}
\email{zhengweiliu@fas.harvard.edu}

\begin{document}
\maketitle
\begin{abstract}
We give an algebraic formulation of our pictorial proof of the reflection positivity property for  Hamiltonians.
We apply our methods to the widely-studied Levin-Wen models and prove the reflection positivity property for a natural class of those Hamiltonians, both with respect to vacuum and to bulk excitations.
\end{abstract}

The reflection positivity property has played a central role in both mathematics and physics, as well as providing a crucial link between the two subjects. In a previous paper we gave a new geometric approach to understanding reflection positivity in terms of pictures. Here we give a transparent algebraic formulation of our pictorial approach. 
We use insights from this translation to establish the reflection positivity property for the fashionable Levin-Wen models.with respect both to vacuum and to bulk excitations. We believe these methods will be useful for understanding a variety of other problems.

\thispagestyle{empty}
\section{Introduction}
In an earlier paper~\cite{JafLiu17}, we gave a new proof of the reflection-positivity (RP) property for Hamiltonians, see Definition~\ref{Def: RP}.   We presented that  proof within the framework of a picture language~\cite{JL-Pic}. Our  language includes a geometric transformation $\FS$, that we call the string Fourier transform (SFT).  The SFT acts on pictures by rotation, and it generalizes the usual Fourier transform that acts on functions, see~\cite{JafLiu17}.  

The picture approach has a great  advantage: we find it very intuitive, illustrating the generality and geometric nature of RP.  But it also has a disadvantage, especially for readers unfamiliar with picture language: it could appear to the uninitiated as a difficult proof  to understand.  

In this paper we elaborate our previous work in two ways.  Firstly we translate our picture proof in~\cite{JafLiu17} into an algebraic proof. We begin with an algebraic formulation of $\FS$ in Definition \ref{Def: FS}. In the remainder of \S \ref{Sec: Alg RP} we prove a general theorem  about RP.  We hope that this exercise makes our pictorial proof accessible for any reader who compares the two methods. Moreover we believe that it should make clear why we find our pictorial method of proof both attractive and transparent. 

Secondly we take advantage of the generality of our pictorial method to analyze some other pictures that occur in the theoretical physics literature. 
Levin and Wen introduced a set of lattice models to study topological order~\cite{LevWen05}. These models generalize the  $\mathbb{Z}_{2}$ toric code of Kitaev~\cite{K06}; for  background see Kitaev and Kong~\cite{KitKon12}.  Levin and Wen showed that ground states of their models correspond to topological quantum field theories in the sense of Turaev and Viro~\cite{TurVir}.  In their paper, Kitaev and Kong give an interesting dictionary to translate between these two sets of concepts.  

 In \S \ref{Sec: LW} we study Levin-Wen models for graphs on surfaces, using the data of unitary fusion categories. We then use our new methods to establish Theorem~\ref{Thm:RP-LW}, the main new result in this paper:  Levin-Wen Hamiltonians have the RP property.  Although we do not analyze it in detail, our method also proves the RP property for higher-dimensional pictorial pictorial models, such as the Walker-Wang models~\cite{WW12}. 

\subsection{The Framework of our RP Proof}
We gave our pictorial proof in \cite{JafLiu17} within the framework of subfactor planar para algebras.  For background see~\cite{JafPed15,Liuex,JafJan16,JafJan17} and  the extensive citations that these papers contain to  work on RP by Osterwalder, Schrader, Biskup, Brydges, Dyson, Frank, Fr\"ohlich,    Israel,  J\"akel, Jorgensen,  Klein, Landau, Lieb, Macris,  Nachtergaele, Neeb, Olafsson, Seiler, Simon, Spencer,  and others.  

A novel aspect of the proof of RP in~\cite{JafLiu17}  was our observation that the positivity of the string Fourier transform $\FS(-H)$ of $H$ ensures the RP property. 
In fact when $H$ is reflection-invariant,  the positivity of $\FS(-H_0)$ is sufficient to ensure RP for ~$H$, where $H_0$  denotes the part of  $H$ that maps across  the reflection mirror. 

In \S \ref{Sec: Alg RP}, we present algebraic definitions of $\FS$, of  the convolution product~$*$, and of the RP property.  While this  may appear somewhat different from the standard definitions,  one can recover the results in \cite{JafLiu17} by a proper choice of the Hilbert space and the Hamiltonian. We do not pursue this comparison in this paper. We attempt  to make minimal assumptions in our statements, so that the methods here could be applied in a wide variety of circumstances.

\subsection{Our Example}
In \S \ref{Sec: LW} we consider the Levin-Wen model on a surface which has a reflection mirror. 
The Hamiltonian is an action on the Hilbert space: it is the sum of contributions from Wilson loops on plaquettes and actions on sites. The terms in $H$ arising from the actions on sites do not contribute to $H_{0}$. 
In the Levin-Wen model, $H_0$ is the sum of the actions on plaquettes that cross the reflection mirror.   

When the plaquette $p$ crosses the mirror $P$, we decompose the Wilson loop as a half circle and its mirror image. 
The action of $\FS$ on a picture is to rotate the picture by $90^{\circ}$. Pictorially we can consider the actions of the two half circles after rotation as the product of a half circle and its adjoint, namely its vertical reflection. So the $\FS(H_p)$ should be positive.  The sticking point is that the actions of the two half circles are not independent, as they share boundary conditions on the mirror. So $H_p$ is not simply a tensor product of operators on two sides of the mirror. Technically we need to take care of the boundary condition in the decomposition of $H_0$. Adding the boundary condition to the decomposition, we prove that $\FS(-H_0)$ is positive. 

Combining this work with the statements in \S \ref{Sec: Alg RP}, we obtain our main result. We remark that RP of the Hamiltonian $H$ in the Levin-Wen model on a torus not only works for the expectation in the vacuum state, but also for the expectation in bulk excitations (objects in the Drinfeld center). Each bulk excitation defines its own one-dimensional lower quantized theory that are topologically entangled on the two boundary circles. We expect this realization to be useful in the study of the anomaly theory on the boundary. 

\section{Algebraic Reflection Positivity}\label{Sec: Alg RP}
In this section we look again at results that we proved in~\cite{Liuex,JafLiu17}, using pictorial methods in the general framework of subfactor planar para algebras. Here we give purely algebraic definitions and proofs, in order to ensure that the ideas and the exposition are accessible to readers who are not familiar with picture language.  

Suppose $\Hil_+$ is a finite dimensional Hilbert space and $\Hil_-$ is its dual space. Let $\langle \cdot, \cdot \rangle_{\Hil_{\pm}}$ be the inner product of the Hilbert spaces $\Hil_{\pm}$. Let $\theta$ be the Riesz representation map from $\Hil_\pm$ to $\Hil_{\mp}$. Then for any $x,x'\in \Hil_+$, their inner product is given by
$$\langle x, x' \rangle_{\Hil_{+}}=\langle \theta(x'), \theta(x) \rangle_{\Hil_{-}}\;.$$
Let $\Hil_{-+}=\Hil_{-}\otimes\Hil_{+}$ denote the tensor product  Hilbert space with the induced inner product $\lra{\ , \ }_{\Hil_{-+}}$, and likewise denote $\Hil_{+-}=\Hil_{+}\otimes \Hil_{-}$.

\begin{definition}[\bf Reflection-Positivity Property] \label{Def: RP} The map $H \in  \hom(\Hil_{-+})$ has the RP property, if for any $x', x\in\Hil_{+}$, and any $\beta\geqslant0$, 
$$\langle \theta(x') \otimes x', e^{-\beta H} \theta(x)\otimes x \rangle_{\Hil_{-+}} \geq 0\;.$$
\end{definition}

\begin{definition}[\bf SFT]\label{Def: FS}
The string Fourier transform $\FS: \hom(\Hil_{-+}) \to  \hom(\Hil_{+-})$ is a map such that for $T \in  \hom(\Hil_{-+})$, and for arbitrary $x,x' \in \Hil_+$ and $y, y'\in \Hil_-$, 
\[
\langle  x \otimes y  , \FS(T) (x'\otimes y') \rangle_{\Hil_{+-}}= \langle  \theta(x') \otimes x  , T (y'\otimes \theta(y)) \rangle_{\Hil_{-+}}\;.
\] 
\end{definition}

\begin{remark}[\bf A Key Identity]
Definition \ref{Def: FS}, with $T=e^{-\beta H}$,  $x=x'$, and $y=y'=\theta(\tilde x)$, and  substituting $x$ for $\tilde x$, yields  
\beq\label{RP-Identity}
&&\hskip-.4in\langle \theta(x') \otimes x', e^{-\beta H} (\theta( x)\otimes  x ) \rangle_{\Hil_{-+}} \nonumber\\
&&=\langle  x' \otimes \theta( x),   \FS(e^{-\beta H}) (x' \otimes \theta( x)) \rangle _{\Hil_{+-}}\;.
\eeq
Thus the RP property for $H$ is equivalent to the positivity of the expectation of $\FS(e^{-\beta H})$ in vectors that are tensor products.
\end{remark}

\begin{theorem}[\bf First RP Statement]\label{Thm: RP1}
A transformation $H \in  \hom(\Hil_{-+})$ satisfying $\FS(-H) \geq 0$ has the RP property. 
\end{theorem}

The map  $\theta$ defines  a map from $\hom (\Hil_\pm)$ to $\hom(\Hil_{\mp})$. For $H'_{\pm}\in\hom(\Hil_{\pm})$ let  
\[
\theta(H'_\pm):=\theta H'_\pm \theta\;.
\]
Extend the definition of $\theta$ as an anti-linear map on $\Hil_{-+}$:
For any $y \otimes x \in \Hil_{-+}$,  let
\[
\theta(y \otimes x):=\theta(x) \otimes \theta(y)\in\Hil_{-+}\;.
\]
Thus
	\be
	\lra{\theta(y\otimes x), \theta(y'\otimes x')}_{\Hil_{-+}}
	= \lra{y'\otimes x', y\otimes x}_{\Hil_{-+}}\;.
	\ee

A more detailed condition on $H$ that yields  the RP property depends (as in past studies) on properties of  the part of $H$ mapping between $\Hil_{+}$ and $\Hil_{-}$.    
For $H\in \hom(\Hil_{-+})$,  let $\theta(H):=\theta H \theta\in\hom(\Hil_{-+})$.

\begin{theorem}[\bf Second RP Statement]\label{Thm: RP2}
Suppose
$$H=H_-+H_0+H_++\lambda I,$$
where $\lambda \in \mathbb{R}$, $H_+ =I_- \otimes H'_+$, for some $H'_+ \in \hom (\Hil_+)$, and where $\theta(H_+)=H_-$.
If $\FS(-H_0)\geq 0$, then $H$ has the RP property. 
\end{theorem}

\subsection{Algebraic Properties of the SFT}  In this section we establish  algebraic properties of $\FS$. We use them in the next section to prove  Theorem~\ref{Thm: RP1} and Theorem~\ref{Thm: RP2}.  

\begin{proposition}\label{Prop:SFT-I}
The SFT of the identity is non-negative, $$\FS(I)\geq0\;.$$ 
\end{proposition}

\begin{proof}
Let $\{x_{i}\}$ denote an orthonormal basis for $\Hil_{+}$ and $\{y_{i}\}=\{\theta(x_{i})\}$ an orthonormal basis for $\Hil_{-}$.  A vector $w\in\Hil_{+-}$ has an expansion $w=\sum_{ij}w_{ij}\,x_{i}\otimes \theta(x_{j})$.  According to Definition~\ref{Def: FS}, 	\beqs
	&&\hskip-.6in\lra{w, \FS(I)w}_{\Hil_{+-}} \\
	&=& \sum_{i,j,i',j'} \overline{w_{ij}} w_{i'j'}
	\lra{\theta(x_{i'})\otimes x_{i}   ,  \theta( x_{j'})\otimes x_{j}}_{\Hil_{-+}}\\
	&=& \abs{ \sum_{i,j} \overline{w_{ij} }\lra{x_{i}, x_{j}}_{\Hil_{+}}}^{2}
	= \abs{ \sum_{i} w_{ii}}^{2}\geq0\;,
	\eeqs
showing an arbitrary expectation of  $\FS(I)\geq0$.
\end{proof}

\begin{remark}
The RP property Definition \ref{Def: RP},  for the case  $H=0$, is a special example of an expectation of $\FS(I)$, namely 
\beqs
&&\hskip -.7in \lra{\theta(x')\otimes x',  \theta(x)\otimes x}_{\Hil_{-+}}\\
&=&   \lra{ x'\otimes \theta(x)  ,   \FS(I) (x'\otimes \theta(x))   }_{\Hil_{+-}}\\
&=&\abs{\lra{x',x}_{\Hil_{+}}}^{2} \geq0\;.
\eeqs
\end{remark}

\begin{remark}
The transformation $\FS^{-1}:  \hom\Hil_{+-} \to \hom\Hil_{-+}$ is,
\[
\langle  y \otimes x  , \FS^{-1}(S) (y'\otimes x') \rangle_{\Hil_{-+}}
= \langle  x \otimes \theta(x')  , S (\theta(y)\otimes y') \rangle_{\Hil_{+-}}\;.
\]
\end{remark}

\begin{proposition}
For any $T \in \hom(\Hil_{-+})$,
$$\FS(\theta(T))=\FS(T)^*.$$
\end{proposition}

\begin{proof}
For any $x,x' \in \Hil_{+}$ and $y, y' \in \Hil_{-}$, the matrix elements of $\FS(\theta(T))$ are
\beqs
&&\hskip -.15in\langle  x \otimes y  , \FS(\theta(T)) (x'\otimes y') \rangle_{\Hil_{+-}}
= \langle  \theta(x') \otimes x  , \theta(T) (y'\otimes \theta(y)) \rangle_{\Hil_{-+}} \\
&&
= \langle  \theta(x') \otimes x  , \theta T \theta (y'\otimes \theta(y)) \rangle_{\Hil_{-+}} 
 =\langle  T \theta (y'\otimes \theta(y)), \theta  (\theta(x') \otimes x) \rangle_{\Hil_{-+}}  \\
&& =\langle  T (y\otimes \theta(y')),  \theta(x) \otimes x' \rangle_{\Hil_{-+}}   
 =\overline{\langle \theta(x) \otimes x', T (y\otimes \theta(y')) \rangle}_{\Hil_{-+}}  \\
&& =\overline{\langle  x' \otimes y'  , \FS(T) (x\otimes y) \rangle}_{\Hil_{+-}} 
 = \langle  x \otimes y  , \FS(T)^* (x'\otimes y') \rangle_{\Hil_{+-}}\;.
\eeqs
Thus the matrix elements agree as claimed.
\end{proof}

\begin{corollary}
A Hamiltonian $H\in\hom(\Hil_{-+})$ is reflection invariant, iff its SFT is hermitian on $\Hil_{+-}$. In other words,
\[
\theta(H)=H   \quad \Longleftrightarrow \quad \FS(H)=\FS(H)^{*}\;.
\]
\end{corollary}

\begin{remark}
Pictorially we represent $\theta$ in~\cite{JafLiu17} as a horizontal reflection, $*$ as a vertical reflection, and $\FS$ as a clockwise $90^\circ$ rotation. 
\end{remark}

\begin{definition}
Let $Y: \Hil_{+-} \otimes  \Hil_{+-} \to    \Hil_{+-}$ be given by 
$$Y(x_{1} \otimes y_{1} \otimes x_{2} \otimes y_{2}):= \langle \theta(y_{1}),x_{2} \rangle_{\Hil_{+}} \,x_{1} \otimes y_{2}\;,$$
for any $x_{1} \otimes y_{1} \otimes x_{2} \otimes y_{2} \in
\Hil_{+-} \otimes  \Hil_{+-}$.
\end{definition}

\begin{lemma}\label{Lem: Y}
Let $\mathcal{B}$ be an orthonormal basis of $\Hil_+$. Then for any $x \in \Hil_+$ and $y \in \Hil_-$,
$$Y^*(x \otimes y) = \sum_{\beta \in \mathcal{B}} x \otimes \theta(\beta) \otimes \beta \otimes y\;.$$ 
Also $$YY^{*}=\dim(\Hil_{+})\, I \;,\text{ on }\Hil_{+}\otimes \Hil_{-}\;.$$

\end{lemma}
\begin{proof}
For any $x,x_{1},x_{2}\in \Hil_+$ and $y,y_{1},y_{2} \in \Hil_-$, and with $\Hil_{+-}^{2}=\Hil_{+-}\otimes \Hil_{+-}$, 
\begin{align*}
&\hskip-.8in\sum_{\beta \in \mathcal{B}} \langle x \otimes \theta(\beta) \otimes \beta \otimes y, x_{1} \otimes y_{1} \otimes x_{2} \otimes y_{2} \rangle_{\Hil_{+-}^{2} }\\
=&\sum_{\beta \in \mathcal{B}}  \langle \theta(\beta), y_{1} \rangle_{\Hil_{-}} 
\langle \beta, x_{2} \rangle_{\Hil_{+}} 
\langle x\otimes y, x_{1} \otimes y_{2} \rangle_{\Hil_{+-}} \\
=&\sum_{\beta \in \mathcal{B}}  
\langle \theta(y_{1}), \beta \rangle_{\Hil_{+}} 
\langle \beta, x_{2} \rangle_{\Hil_{+}} 
\langle x\otimes y, x_{1} \otimes y_{2} \rangle_{\Hil_{+-}} \\
=&\sum_{\beta \in \mathcal{B}}  
\langle \theta(y_{1}), x_{2} \rangle_{\Hil_{+}}  
\langle x\otimes y, x_{1} \otimes y_{2} \rangle_{\Hil_{+-}}  \\
=& \langle x\otimes y, Y (x_{1} \otimes y_{1} \otimes x_{2} \otimes y_{2} ) \rangle_{\Hil_{+-}} \\
=& \langle Y^{*}(x\otimes y), x_{1} \otimes y_{1} \otimes x_{2} \otimes y_{2}  \rangle_{\Hil_{+-}^{2} }
\;. 
\end{align*}
This completes the computation of $Y^{*}$.  Also  
	\[
	YY^{*} x\otimes y
	= Y \sum_{\beta\in \mathcal{B}} x \otimes \theta(\beta) \otimes \beta \otimes y
	= \left(\sum_{\beta\in \mathcal{B}} \langle \beta, \beta\rangle \right) x\otimes y\;.
	\]
Note that the $\beta$ are an orthonormal basis for $\Hil_{+}$, so the sum in parentheses equals   $\dim (\Hil_{+})$.
\end{proof}

\begin{definition}[\bf Convolution]\label{Def: convolution}
For  $A,B \in \hom(\Hil_{+-})$, their convolution product is    
$
A*B:=Y (A\otimes B) Y^* 
$.
\end{definition}

The convolution is associative, as a consequence of Lemma~\ref{Lem: Y}. 
\begin{remark}
Let  $\mathcal{B}$ be an orthonormal basis for $\Hil_{+}$ and $\theta(\mathcal{B})$ a corresponding basis for $\Hil_{-}$.  Then for $i,j\in\mathcal{B}$, the vectors $i\otimes\theta(j)$ are an orthonormal basis for $\Hil_{+-}$.  
A matrix unit $E_{ii'jj'}\in \hom(\Hil_{+-})$ is zero except  on $i'\otimes \theta(j')$ and maps that vector to the vector $i\otimes \theta(j)$.  The transformations $A,B\in \hom(\Hil_{+-})$ can be written 
\[
A = \sum_{i,i',j,j'\in\mathcal{B}}  a_{ii'jj'}  E_{ii'jj'}\;,\ \ 
B= \sum_{k,k',\ell,\ell'\in\mathcal{B}}  b_{kk'\ell\ell'}  E_{kk'\ell\ell'}\;.
 \]
One can compare the matrix elements of $AB$ with those of $A*B$, namely 
\[
	\lra{\alpha\otimes\theta(\beta),    
	(A B) \alpha'\otimes\theta(\beta')}_{\Hil_{+-}}
	=  \sum_{k,k'\in\mathcal{B}} a_{\alpha k \beta k'}\, b_{k \alpha' k' \beta'} \;,
\]
\[
	\lra{\alpha\otimes\theta(\beta),    
	(A* B) \alpha'\otimes\theta(\beta')}_{\Hil_{+-}}
	=  \sum_{k,k'\in\mathcal{B}} a_{\alpha\alpha' kk'}\, b_{kk'\beta\beta'} \;.
\]
In particular on $\Hil_{+-}$, one has  $I = \sum_{ij} E_{iijj}$ and 
\be
I*I=\dim(\Hil_{+})I\;.
\ee
In~\cite{JafLiu17} we represent $A$ and $B$ pictorially as ``two-box'' pictures. The multiplication  $AB$ is given by vertical composition of the two-box pictures, while the multiplication $A*B$ is given by the corresponding  horizontal composition of the same pictures.  
\end{remark}

\begin{theorem}[\bf SFT on Products]\label{Prop: FS MC}
The SFT maps products in $\hom(\Hil_{-+})$ to convolutions in $\hom(\Hil_{+-})$. For $S,T \in \hom(\Hil_{-+})$, 
$$\FS(ST)=\FS(S)*\FS(T)\;.$$
\end{theorem}

\begin{proof}
Let $x_{1},x_{2} \in \Hil_+$ and $y_{1}, y_{2} \in \Hil_-$. By Definition \ref{Def: FS},
\begin{align*}
&\hskip -.15in\langle  x_{1} \otimes y_{1}  , \FS(ST) (x_{2}\otimes y_{2}) \rangle_{\Hil_{+-} }
= \langle  \theta(x_{2}) \otimes x_{1}  , ST (y_{2}\otimes \theta(y_{1})) \rangle_{\Hil_{-+} }\\
=&\sum_{\beta_1,\beta_2 \in \mathcal{B}} 
\langle  \theta(x_{2}) \otimes x_{1}, S  (\theta(\beta_1) \otimes \beta_2) \rangle_{\Hil_{-+} } 
\langle \theta(\beta_1) \otimes \beta_2, T (y_{2}\otimes \theta(y_{1})) \rangle_{\Hil_{-+} } \\
=&\sum_{\beta_1,\beta_2 \in \mathcal{B}} 
\langle  x_{1} \otimes \theta(\beta_2)  , \FS(S)  (x_{2} \otimes \theta(\beta_1)) \rangle_{\Hil_{+-} } 
\langle \beta_2 \otimes y_{1} , \FS(T) (\beta_1 \otimes y_{2}) \rangle_{\Hil_{+-} } \\
=&\sum_{\beta_1,\beta_2 \in \mathcal{B}} 
\langle  x_{1} \otimes \theta(\beta_2)  \otimes \beta_2 \otimes y_{1} , 
(\FS(S) \otimes \FS(T))  (x_{2} \otimes \theta(\beta_1) \otimes \beta_1 \otimes y_{2} ) \rangle_{\Hil_{+-}^{2} } \\
=&\langle Y^* (x_{1} \otimes y_{1})  ,  (\FS(S) \otimes \FS(T)) Y^* (x_{2}\otimes y_{2}) \rangle_{\Hil_{+-}^{2} }\\
=&\langle x_{1} \otimes y_{1} , Y (\FS(S) \otimes \FS(T)) Y^* (x_{2}\otimes y_{2}) \rangle_{\Hil_{+-}} \\
=&\langle  x_{1} \otimes y_{1}  , \FS(S)*\FS(T) (x_{2}\otimes y_{2}) \rangle_{\Hil_{+-}}\;,
\end{align*}
where we infer the last three equalities from Lemma~\ref{Lem: Y} and Definition~\ref{Def: FS}. Therefore, the operators agree as claimed.
\end{proof}

\begin{theorem}[\bf Schur Product Theorem]\label{Thm: Schur Product}  Let $S, T\in  \hom(\Hil_{+-})$.
If $S\geq 0$ and $T \geq 0$, then $S*T \geq 0$.
\end{theorem}

\begin{proof}
Let  $\sqrt{S}$ and $\sqrt{T}$ denote the positive square roots of $S$ and $T$. 
By Definition \ref{Def: convolution}, one has 
$S*T=(Y (\sqrt{S} \otimes \sqrt{T})) (Y (\sqrt{S} \otimes \sqrt{T}))^* \geq 0$.
\end{proof}

\begin{corollary}[\bf Exponentials and Products]\label{Cor:Positivity}
If $\FS(S)\geq0$, then $\FS(e^{S})\geq0$.   
If $\FS(S)\geq0$ and $\FS(T)\geq0$, then $\FS(ST)\geq0$.
\end{corollary}

\begin{proof}
From Theorem \ref{Prop: FS MC}, 
$
\FS(ST) = \FS(S)*\FS(T)\;.
$
We then infer $\FS(ST)\geq0$ from Theorem \ref{Thm: Schur Product}.  Likewise $\FS(S)\geq0$ ensures $\FS(S^{n})\geq0$ for any natural number $n$.  Since $\FS$ is a linear transformation, and the exponential power series has positive coefficients, so $\FS(e^{S}-I)\geq0$.  But from  Proposition \ref{Prop:SFT-I} we know  $\FS(I)\geq0$, hence $\FS(e^{S})\geq0$.
\end{proof}

\begin{proposition}[\bf A Positivity Property]\label{Prop: FS TT}
If $T_+ \in \hom(\Hil_+)$, then
\[
\FS(\theta(T_+)\otimes T_+)\geq 0\;.
\]
\end{proposition}

\begin{proof}
Let $\{x_{i}\}$ be an orthonormal basis for $\Hil_{+}$ and $\{y_{j}\}$ an orthonormal basis for $\Hil_{-}$.  Let $s_{ij}=\lra{x_{i}, T\theta(y_{j})}_{\Hil_{+}}$.   A vector $a\in\Hil_{+-}$ has the form $a=\sum_{i,j}a_{ij}\,x_{i}\otimes y_{j}$.  According to Definition~\ref{Def: FS}, the matrix elements of $\FS(\theta(T_{+})\otimes T_{+})$ on $\Hil_{+-}$ in the basis $x_{i} \otimes y_{j}$ are 
	\begin{eqnarray*}
	&&\hskip -.5in \lra{x_{i}\otimes y_{j}, \FS(\theta(T_{+})\otimes T_{+})(x_{i'}\otimes y_{j'})}_{\Hil_{+-}}\\
	&& = \lra{\theta(x_{i'})\otimes x_{i} ,  (\theta(T_{+})\otimes T_{+})
	(y_{j'}\otimes \theta(y_{j}))}_{\Hil_{-+}}\\
	&& = \lra{\theta(x_{i'}), \theta(T_{+})y_{j'}}_{\Hil_{-}}
	\lra{x_{i}, T_{+}\theta(y_{j})}_{\Hil_{+}}\\
	&& = \lra{T_{+}\theta(y_{j'}), x_{i'}}_{\Hil_{+}}
	\lra{x_{i}, T_{+}\theta(y_{j})}_{\Hil_{+}}
	= \overline{s_{i'j'}} \,s_{ij}\;.
	\end{eqnarray*}
Thus
\beqs
\lra{a, \FS(\theta(T_{+})\otimes T_{+})a}_{\Hil_{+-}}
&=& \sum_{i,j,i',j'}  \overline{a_{ij}} a_{i'j'}    \overline{s_{i'j'}} \,s_{ij} \\
&=&\left |\sum_{i,j}\overline{a_{ij}}\,s_{ij}\right|^{2}
\geq0\;,
\eeqs
to complete the proof.
\end{proof}

\subsection{Proof of the RP Property} 
We apply the above properties of $\FS$ to establish the reflection positivity property for $H$.
\begin{proof}[\bf Proof of Theorem \ref{Thm: RP1}]
We assume $\FS(-H)\geq0$, so Corollary \ref{Cor:Positivity} and $\beta\geq0$ ensures  $\FS(e^{-\beta H}) \geq 0$.  Hence \eqref{RP-Identity} is the expectation of a positive operator, which  establishes the RP property for $H$.   
\end{proof}

\begin{proof}[\bf Proof of Theorem \ref{Thm: RP2}]
See also the proof of Theorem~4.2 in~\cite{JafJan16}.
Assume $\FS(H_{0})\geq0$. For $s>0$, define 
	\beqs
	-H(s) &=& -H_0 +  s(H_{-} - s^{-1}I)(H_{+}  - s^{-1}I) \\
	&=& -H_{0} +   s\,\theta(T_+)\otimes T_{+}\;.
	\eeqs  
Here $T_+=H_{+} - s^{-1}I$.  As $H_{+}=I_- \otimes H'_{+}$  acts on $\Hil_{+}$, we infer that  $T_{+}$ satisfies the hypotheses of Proposition~\ref{Prop: FS TT}.  Hence 
$\FS(\theta(T_+)\otimes T_+) \geq 0$, and consequently $\FS(-H(s))\geq 0$.  We then conclude from Theorem~\ref{Thm: RP1}, that $H(s)$ has the RP property. Adding a constant to $H(s)$ does not affect RP, so
	$
	H(s)+(\lambda+s^{-1})I=H-s\,\theta(H_{+})H_{+}
	$ 
also has the RP property. Namely for all $x', x\in \Hil_{+}$ and all $\beta\geq 0$, 
$$\langle \theta(x') \otimes x', e^{-\beta H +s\,\beta \theta(H_{+})H_{+}} \theta(x)\otimes x \rangle_{\Hil_{-+}} \geq 0\;.$$
This representation is continuous  in $s$, also at $s=0$. So let $s \to 0+$ to ensure the RP property for $H$. 
\end{proof}

\section{Levin-Wen models}\label{Sec: LW}
In this section, we define the Levin-Wen model for graphs in surfaces using the data of unitary fusion categories. 
Our main result is proving reflection positivity for the Hamiltonian in the Levin-Wen model.

\subsection{Graphs in surfaces}
Let $M_+$ be a surface in the half space $\mathbb{R}^3_+=\{(x_1,x_2,x_3) | x_1 \geq 0\}$ with boundary $\partial M$ on the plane $P=\{(x_1,x_2,x_3)| x_1=0\}$. Let $\Gamma_+$ be an oriented graph embedded in the surface $M_+$, such that $\Gamma_+ \cap \partial M_+ =\partial \Gamma_+$, namely the boundary points of $\Gamma_+$.
Let $\theta_P$ be the reflection by the hyperplane $P$. Take $M_-=\theta_P(M_+)$. Then $\partial M_-= \partial M_+$. Let $M= M_+ \cup M_-$.
Take $\Gamma_-=\theta_P(\Gamma_+)$, and the orientation is reversed by $\theta_P$. 
Then $\partial \Gamma_-= \partial \Gamma_+$. Take $\Gamma= \Gamma_+ \cup \Gamma_-$.
Then $M$ is a closed surface and $\Gamma$ is a closed oriented graph in $M$.

Denote $E_+=E(\Gamma_+)$ to be the edges of $\Gamma_+$ and $V_+=V(\Gamma_+)$ to be the vertices of $\Gamma_+$. (The boundary points in $\partial \Gamma_+$ are not vertices of $\Gamma_+$.)
Similarly  define $E_-=E(\Gamma_-)$, $E=E(\Gamma)$, $V_-=V(\Gamma_-)$ and $V=V(\Gamma)$.
Then $V=V_+ \cup V_-$ and $V_+\cap V_-=\emptyset$.
Take $E_0=\{e \in E | e\cap P \neq \emptyset\}$, the set of edges go across the plane $P$. 
Then for any $e \in E_0$, its positive half is an edge in $E_+$ and its negative half is an edge in $E_-$. We identify the three edges as the same edge. Then $E_+ \cap E_-=E_0$, $E_+ \cup E_-=E$.

Let $s, t: E\to V$ be the source function and the target function.
For any edge $e\in E$, the end points of $e$ are  $\partial e=\{s(e), t(e)\}$ .
Since the orientation is reversed by $\theta_P$, we have 
$$s(\theta_P(e))=\theta_P \,t(e).$$

For any vertex $v \in V $, we define the set of adjacent edges $E(v)=\{e \in E | v \in \partial e\}$. The cardinality of $E(v)$ is called the degree of the vertex $v$, denoted by $|v|$.
Let $\kappa_v$ be an bijection from $\{1,2,\ldots, |v|\}$ to $E(V)$, so that the numbers go from $1$ to $|v|$ anti-clockwise around the vertices. The order $\kappa_v$ is determined by the choice of the edge $\kappa_v(1)$.
Define $\varepsilon_v(e)=+$ if $s(e)=v$;  $\varepsilon_v(e)=-$ if $t(e)=v$.

\subsection{Unitary fusion categories}
Suppose $\C$ is a unitary fusion category, (corresponding to a unitary tensor category in \cite{KitKon12}).
Let $\Irr$ be the set of irreducible objects (i.e., simple objects) of $\C$, and let $1\in Irr$ be the trivial object. Take $A=\bigoplus_{X \in \Irr} X$ and $A^n:=\otimes_{k=1}^n A$.  
For any object $X$, let $\ONB(X)$ denote an ortho-nomal basis of $\hom_{\C} (1, X)$. Let $d(X)$ be the quantum dimension of $X$. Let $1_X$ be the identity map in $\hom_{\C}(X,X)$.
Define $X^+=X$ and $X^-$ to be the dual object of $X$.
For any objects $X,Y,Z$ in $\C$,
let $\theta_{\C}: \hom_{\C}(X \otimes Y, Z) \to \hom_{\C}(Y^- \otimes X^-, Z^-) $ be the modular conjugation on $\C$. Pictorially $\theta_{\C}$ is a horizontal reflection.

Let $\cap_A$ be the co-evaluation map from $1$ to $A^2$ and $\cup_A$ be the evaluation map from $A^2$ to $1$. Then $\cup_A \cap_A =d(A)$ and $(1_A \otimes \cup_A) (\cap_A \otimes 1_{A})=1_A$.
Define $\rho: \hom_{\C} (1, A^n) \to \hom_{\C} (1, A^n)$: for  $x \in \hom_{\C} (1, A^n)$, let 
$$\rho(x)=(\cup_A  \otimes 1_{A^n}) (1_{A} \otimes x \otimes 1_A) \cap_A.$$
Pictorially, we represent $x$ as 
\raisebox{-.35cm}{\begin{tikzpicture}
\node at (0,0) {$\bullet$};
\node at (-.25,0) {$x$};
\draw (0,0)--(-.5,-.5);
\draw (0,0)--(-.25,-.5);
\draw (0,0)--(.25,-.5);
\draw (0,0)--(.5,-.5);
\node at (0,-.5) {$...$};
\end{tikzpicture}},
where the $n$ edges are all labelled by the object $A$.
Then 
\[
\rho(x):=
\raisebox{-.55cm}{\begin{tikzpicture}
\node at (0,0) {$\bullet$};
\node at (-.25,0) {$x$};
\draw (0,0)--(-.5,-.5) arc (0:-180:.25) arc (180:0:1)--++(0,-.5);
\draw (0,0)--(-.25,-.5)--++(0,-.5);
\draw (0,0)--(.25,-.5)--++(0,-.5);
\draw (0,0)--(.5,-.5)--++(0,-.5);
\node at (0,-.5) {$...$};
\end{tikzpicture}}\;.
\]
For any $y,z\in \hom_{\C}(A^2,A)$, define $C_{y,z}: \hom_{\C} (1, A^n) \to \hom_{\C} (1, A^n)$: for any $x \in \hom_{\C} (1, A^n)$, $n\geq 2$, take the algebraic expression to be 
$$C_{y,z}(x):=(y  \otimes 1_{A^{n-2}} \otimes z) (1_{A} \otimes x \otimes 1_A) \cap_A.$$
The corresponding pictorial representation is, 
\be\label{Vertex}
C_{y,z}(x)=
\raisebox{-.55cm}
{\begin{tikzpicture}
\node at (0,0) {$\bullet$};
\node at (-.25,0) {$x$};
\draw (0,0)--(-.5,-.5)--(-.75,-.75); 
\draw (-1,-.5)--(-.75,-.75)--++(0,-.25);
\node at (-1,-.75) {$y$};
\draw (-1,-.5) arc (180:0:1);
\draw (0,0)--(-.25,-.5)--++(0,-.5);
\draw (0,0)--(.25,-.5)--++(0,-.5);
\node at (0,-.5) {$...$};
\draw (0,0)--(.5,-.5)--(.75,-.75);
\draw (1,-.5)--(.75,-.75)--++(0,-.25);
\node at (.5,-.75) {$z$};
\end{tikzpicture}}
\;.
\ee

\subsection{Configuration spaces}
For every edge $e \in E$,  we define $\Hil_e=L^2(Irr)$.  Moreover, the delta functions $\delta_j$, $j \in \Irr$, form an ONB of $L^2(Irr)$.   
For every vertex $v \in V$,  we define $\Hil_v=\hom_{\C} (1, A^{|v|})$.    
\begin{definition}[\bf LW Hilbert spaces] \label{Def:LWHilbertSpaces}
Define the Hilbert spaces for the Levin-Wen model as 
\begin{align*}
\Hil_+&:=(\bigotimes_{v \in V_+} \Hil_v) \otimes (\bigotimes_{e \in E_+}  \Hil_e)\;,\\
\Hil_-&:=(\bigotimes_{v \in V_-} \Hil_v) \otimes (\bigotimes_{e \in E_-}  \Hil_e) \;, \\
\Hil&:=(\bigotimes_{v \in V} \Hil_v) \otimes (\bigotimes_{e \in E}  \Hil_e) .
\end{align*}
\end{definition}
The two Hilbert spaces $\Hil_-$ and $\Hil_+$ are dual to each other with respect to the Riesz representation $\theta$.
Define the embedding map 
\be\label{DefIota}
\iota: \Hil \to \Hil_{-+}=\Hil_{-} \boxtimes \Hil_{+}\;,
\ee 
as a multilinear extension of the map on an ONB:
\beqs
&&\hskip -2.3cm
\iota \left((\bigotimes_{v \in V} \beta_v) \otimes (\bigotimes_{e \in E}  \delta_{j(e)}) \right)\nonumber\\
&=&
(\bigotimes_{v \in V_-} \beta_v) \otimes (\bigotimes_{e \in E_-}  \delta_{j(e)}) \boxtimes 
(\bigotimes_{v \in V_+} \beta_v) \otimes (\bigotimes_{e \in E_+}  \delta_{j(e)}), 
\eeqs
for any $\beta_v \in ONB(\Hil_v)$ and any $j(e)\in \Irr$.
Extend the reflection $\theta_P$ to an anti-unitary $\theta:\Hil_+ \to \Hil_-$ as follows,
\begin{align*}
\theta\left((\bigotimes_{v \in V_+} \beta_v) \otimes (\bigotimes_{e \in E_+}  \delta_{j(e)})\right)
&=\bigotimes_{v \in V_-} \theta_{\C}(\beta_{\theta_{P}(v)}) \otimes (\bigotimes_{e \in E_+}  \delta_{j(\theta_{P}(e))})\;.
\end{align*}

\subsection{Hamiltonians}
Let $Irr^n$ denote the tensor product,
\[
Irr^n:=\{j_1 \otimes j_2 \otimes  \cdots \otimes  j_n | j_k \in Irr, 1\leq k \leq n \}\;.
\]
Define $P_{v, \vec{j} } $ to be the projection from $\hom_{\C}(1, A^{|v|})$ on to $\hom_{\C}(1, \vec{j} )$ at the vertex $v$.
Define $P_{e, j}$ to be the projection from $L^2(Irr)$ on to $\mathbb{C}\delta_{j}$ at the edge $e$.
For any $v \in V$, the action on the vertex is given by the operator $H_v$ on $\Hil$:
$$H_v=\sum_{\vec{j} \in Irr^{|v|} } P_{v, \vec{j}}    \prod_{k=1}^{|v|} P_{\kappa_v(k), j_k^{\varepsilon_v\kappa_v(k)}}\;.$$

One calls each connected component of $M\setminus \Gamma$ a plaquette. 
Let $\Pla$ be the set of plaquettes. For any $p \in \Pla$, let us denote the vertices and edges on $\partial p$ by $v_1, e_1, v_2, e_2 \ldots, v_m, e_m$ clockwise. 
For any $j \in \Irr$, the action on the plaquette is given by the operator $H_{p,j}$ on $\Hil$:
\begin{align}
H_{p,j}
=&\sum_{\vec{j} \in Irr^{|v|} } \prod_{\ell=1}^{m} P_{e_{\ell}, j_{\ell}}  \left(\sum_{k=1}^m \sum_{j'_k \in Irr} 
\sum_{y_k \in ONB \hom (j \otimes  j_k^{\varepsilon_{v_k}(e_k)},\, 
(j'_{k})^{\varepsilon_{v_k}(e_k)}} \right. \nonumber\\
&\qquad\qquad \left. d(j'_k) \prod_{k=1}^{|v|}   \rho_{v_k}^{1-\kappa_v^{-1}(e_k)}  C_{v_k, y_k, \theta_{\C}(y_{k-1})} \rho_{v_k}^{\kappa_v^{-1}(e_k)-1}  \right),\nonumber
\end{align}
where $y_0=y_n$ and $\rho_{v_k}$, $C_{v_k, y_k, \theta(y_{k-1})}$ are the actions of $\rho$ and $C_{y_k, \theta(y_{k-1})}$ at the vertex $v_k$ respectively.
Here also 
$$H_p= \sum_{j \in Irr} \frac{d(j)^2}{\mu} H_{p,j},$$
where $\displaystyle \mu=\sum_{j \in Irr}d(j)^2$ is the global dimension of $\C$. It is known that $H_p$, for $p\in \Pla$, and $H_v$, for $v\in V$ are mutually commuting projections~\cite{LevWen05,KitKon12}.
In the Levin-Wen model, the Hamiltonian $H$ on $\Hil$ is 
\be\label{LW-Hamiltonian}
H=\lambda_{\Pla}  \sum_{p\in \Pla} (1-H_{p})+ \lambda_{V} \sum_{v \in V} (1-H_v)\;,
\ee
for some $\lambda_{\Pla} \geq 0$ and $\lambda_{V}\geq 0$. 

Pictorially, the action of $H_{p,j}$ is contracting a loop labelled by $j$ in the plaquette $p$ with morphisms in $\C$ on $\partial p$: 
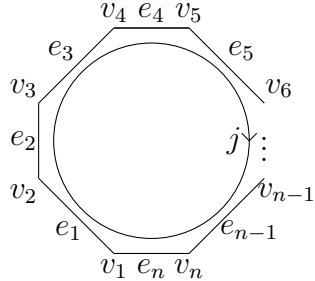
\begin{center}
\begin{tikzpicture}
\begin{scope}[xscale=1,shift={(-.5,0)}]
\draw (0,0) -- (1,0);
\node at (1,0-.2) {$v_n$};
\draw (1,0) --(2,1);
\node at (2+.3,1-.2) {$v_{n-1}$};
\draw (0,3) -- (1,3);
\node at (1,3+.2) {$v_5$};
\draw (1,3)--(2,2);
\node at (2+.2,2+.2) {$v_{6}$};
\node at (2,1.5) {$\vdots$};
\node at (1.5+.3,.5-.2) {$e_{n-1}$};
\node at (1.5+.2,2.5+.2) {$e_{5}$};
\end{scope}

\begin{scope}[xscale=-1, shift={(-.5,0)}]
\draw (0,0) -- (1,0);
\node at (1,0-.2) {$v_1$};
\draw (1,0) --(2,1);
\node at (2+.2,1-.2) {$v_2$};
\node at (2+.2,1.5) {$e_2$};
\draw (0,3) -- (1,3);
\node at (1,3+.2) {$v_4$};
\draw (1,3)--(2,2);
\node at (2+.2,2+.2) {$v_{3}$};
\draw (2,1)--(2,2);
\node at (.5,0-.2) {$e_n$};
\node at (1.5+.1,.5-.2) {$e_1$};
\node at (1.5+.2,2.5+.2) {$e_{3}$};
\node at (0.5,3+.2) {$e_4$};
\end{scope}
\draw (0,1.5) circle (1.3);
\draw (1.2,1.6)--(1.3,1.5)--(1.4,1.6);
\node at (1.1,1.5) {$j$};
\end{tikzpicture}
\end{center}
\noindent
The contraction is induced from the relation 
$$1_{j} \otimes 1_{j_k^{\pm}}=\sum_{j'_k \in \Irr} d(j'_k) \sum_{y_k \in ONB \hom (j \otimes  j_k^{\pm}, j'_k)} y_k^* y_k$$ in $\C$.
See \S 3 of \cite{KitKon12} for more details. Pictorially, this relation changes the shape of a pair of lines labelled by $e_i$ and $j$ and as follows:
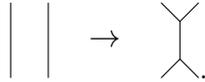
\begin{center}
\begin{tikzpicture}
\draw (0,0)--(0,1);
\draw (.5,0)--(.5,1);
\node at (1.25,.5) {$\to$};
\draw (2,0)--(2.25,.25)--(2.5,0);
\draw (2.25,.25)--(2.25,.75);
\draw (2,1)--(2.25,.75)--(2.5,1);
\end{tikzpicture}.
\end{center}
Then around each vertex $v_i$, the shape of the picture looks like \eqref{Vertex}.

The definition of $H_{p,j}$ is independent of the choice of the starting vertex $v_1$. It is also independent of the order $\kappa_v$.
When we change the orientation of an edge $e$ in the oriented graph $\Gamma$, we replace $P_{e, X}$ by $P_{e, X^-}$. Then the operators $H_v$ and $H_{p,j}$ are not changed. So the operators are essentially independent of the orientation of the graph $\Gamma$.

\subsection{Reflection Positivity}
The main new result of this paper is the following:

\begin{theorem}[\bf RP Property for Levin-Wen Models]\label{Thm:RP-LW}
The Hamiltonian $H$ in \eqref{LW-Hamiltonian},  acting on the Hilbert space~$\Hil$ of Definition \ref{Def:LWHilbertSpaces}, has the RP property: for any $h_+, \Omega_+\in \Hil_{+}$, and $\beta \geq 0$,
\begin{align}
\langle e^{-\beta H} \iota^*(\theta (h_+) \boxtimes h_+),  \iota^*(\theta(\Omega_{+}) \boxtimes \Omega_{+}) \rangle_{\Hil} \geq 0.
\end{align}
\end{theorem}

\begin{lemma}\label{Lem: RP for p}
For any plaquette $p$ across the plane $P$, namely $p\cap P \neq \emptyset$, we have
$\FS (-\iota H_{p,j} \iota^*)\geq 0$.
\end{lemma}

\begin{proof}
For any plaquette $p$ across the plane $P$, 
let us denote the vertices and edges in $\partial p \cap \Gamma_-$ clockwise as $e_0, v_1, e_1 v_2,\ldots, v_m, e_m$; the vertices and edges in $\partial p \cap \Gamma_+$ anti-clockwise as $f_0, w_1, f_1, w_2,\ldots, w_m, f_m$. Then $w_k=\theta (v_k)$, for $1\leq k\leq m$; and $f_k=\theta(e_k)$, for $0 \leq k \leq m$. Moreover, $\varepsilon_{w_1}(f_0)=-\varepsilon_{v_1}(e_0), \varepsilon_{w_n}(f_n)=-\varepsilon_{v_n}(e_n)$. 
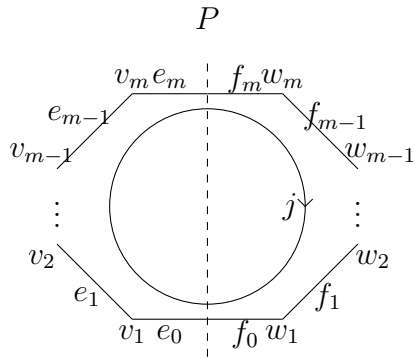
\begin{center}
\begin{tikzpicture}
\draw[dashed] (0,-.5) --(0,3.5); 
\node at (0,4) {$P$};

\draw (0,0) -- (1,0);
\node at (1,0-.2) {$w_1$};
\draw (1,0) --(2,1);
\node at (2+.2,1-.2) {$w_2$};
\draw (0,3) -- (1,3);
\node at (1,3+.2) {$w_m$};
\draw (1,3)--(2,2);
\node at (2+.3,2+.2) {$w_{m-1}$};
\node at (2,1.5) {$\vdots$};
\node at (.5,0-.2) {$f_0$};
\node at (1.5+.1,.5-.2) {$f_1$};
\node at (1.5+.2,2.5+.2) {$f_{m-1}$};
\node at (.5,3+.2) {$f_m$};

\begin{scope}[xscale=-1]
\draw (0,0) -- (1,0);
\node at (1,0-.2) {$v_1$};
\draw (1,0) --(2,1);
\node at (2+.2,1-.2) {$v_2$};
\draw (0,3) -- (1,3);
\node at (1,3+.2) {$v_m$};
\draw (1,3)--(2,2);
\node at (2+.2,2+.2) {$v_{m-1}$};
\node at (2,1.5) {$\vdots$};
\node at (.5,0-.2) {$e_0$};
\node at (1.5+.1,.5-.2) {$e_1$};
\node at (1.5+.2,2.5+.2) {$e_{m-1}$};
\node at (.5,3+.2) {$e_m$};
\end{scope}
\draw (0,1.5) circle (1.3);
\draw (1.2,1.6)--(1.3,1.5)--(1.4,1.6);
\node at (1.1,1.5) {$j$};
\end{tikzpicture}
\end{center}
By the definitions of $H_{p,j}$ and $\iota$, we have 
\begin{align*}
-\iota H_{p,j} \iota^*=&\sum_{j_0, j'_0,j_m,j'_m \in Irr} T_{j_0, j'_0,j_m,j'_m} \boxtimes \theta(T_{j_0, j'_0,j_m,j'_m}),
\end{align*}
where
\begin{align*}
& \hskip -.1in
T_{j_0, j'_0,j_m,j'_m}\\
&=P_{e_0,j_0} P_{e_m, j_m} \\
& \ \ \times\sum_{y_0 \in ONB \hom (j \otimes j_0^{-\varepsilon_{v_1}(e_0)} , (j'_0)^{-\varepsilon_{v_1}(e_0)})} \
\sum_{y_m \in ONB \hom (j \otimes j_0^{\varepsilon_{v_m}(e_m)} ,  (j'_m)^{\varepsilon_{v_m}(e_m)} )} \\
&\qquad \sum_{\vec{j} \in Irr^{|v|-1} } \prod_{\ell=1}^{m-1} P_{e_{\ell}, j_{\ell}}  \left(\sum_{k=1}^{m-1}\, \sum_{j'_k \in Irr} \  \sum_{y_k \in ONB \hom (j \otimes  j_k^{\varepsilon_{v_k}(e_k)}, (j'_k)^{\varepsilon_{v_k}(e_k)}  )} \right. \nonumber\\
&\hskip 1.5in \left. d(j'_k) \prod_{k=1}^{|v|}   \rho_{v_k}^{1-\kappa_v^{-1}(e_k)}  C_{v_k, y_k, \theta_{\C}(y_{k-1})} \rho_{v_k}^{\kappa_v^{-1}(e_k)-1}  \right).
\end{align*}
By Proposition \ref{Prop: FS TT},
$\FS(-\iota H_{p,j} \iota^*) \geq 0$.
\end{proof}

\begin{proof}[\bf Proof of Theorem \ref{Thm:RP-LW}]
Take $\tilde{H}=\iota H \iota^*$.
We have the decomposition 
$\tilde{H}=H_0+H_++H_-+\lambda I$,
such that
\begin{align*}
H_0&=\lambda_{\mathfrak{P}} \sum_{p\in \mathfrak{P}, p\cap P \neq \emptyset} -H_p; \\ 
H_\pm&=\lambda_{\Pla }  \sum_{p\in \Pla, p \subset M_\pm} (1-H_{p})+ \lambda_{V} \sum_{v \in V_\pm} (1-H_v); \\
\lambda&=\lambda_{\mathfrak{P}} \sum_{p\in \mathfrak{P}, p\cap P \neq \emptyset}1.
\end{align*}
Then $\theta(H_+)=H_-$ and $H_+ =I \otimes H'_+$, for some $H'_+ \in \hom (\Hil_+)$.
By Lemma \ref{Lem: RP for p}, $\FS(-H_0)\geq 0$.
By Theorem \ref{Thm: RP2}, $\tilde{H}$ has the RP property. 
For any $h_+, \Omega_+\in \Hil_+, \beta \geq 0,$
\begin{align*}
&\hskip -.5in \langle e^{-\beta H} \iota^*(\theta (h_+) \boxtimes h_+),  \iota^*(\theta(\Omega_{+}) \boxtimes \Omega_{+}) \rangle_{\Hil} \\
=&\langle e^{-\beta \tilde{H}} (\theta (h_+) \boxtimes h_+),  (\theta(\Omega_{+}) \boxtimes \Omega_{+}) \rangle_{\Hil_{-+}} \geq 0.
\end{align*}
Therefore $H$ has the RP property. 
\end{proof}

\subsection{An Interpretation}
Let us explain an elementary example: let  $M_+$  be isotopic to a cylinder, so $M$ is a torus. Take the graph $\Gamma$ to be a square lattice in $M$.  
For the Levin-Wen model on a torus $M$, it is known that the excitations in the bulk are objects of the Drinfeld center $Z(\C)$.
If $\Omega_+$ is the vacuum vector in $\Hil_+$, then $\iota^*(\theta(\Omega_{+}) \boxtimes \Omega_{+})$ is the vacuum vector in $\Hil$, namely all objects and morphisms are trivial. 

We can consider the expectation on the vacuum, namely $\iota^*(\theta(\Omega_{+}) \boxtimes \Omega_{+})$, as a path integral over configurations, where the Hamiltonian acts diagonally. These configurations can be identified as closed string nets on the dual lattice through the modular self-duality proved in \cite{LiuXu}, when $\C$ is a unitary modular tensor category. The RP condition for the path integral in the bulk induces a one-dimensional lower quantum theory on the boundary of $M_+$, which is a union of two circles.

If $\Omega_+$ is an open string with end points on the two boundary circles of $M_+$, then  $\iota^*(\theta(\Omega_{+}) \boxtimes \Omega_{+})$ is a closed string in $M$, corresponding to a bulk excitation. We can still consider the expectation on $\iota^*(\theta(\Omega_{+}) \boxtimes \Omega_{+})$ as a \emph{non-local} path integral. The RP condition for the path integral in the bulk induces a quantum theory topologically entangled on the two boundary circles. As mentioned in the introduction, we expect this realization to be useful in the study of the anomaly theory on the boundary.

\section{Acknowledgement}
{This research in the Mathematical Picture Language Project was supported by the Templeton Religion Trust under grant TRT 0159.}

\newpage
  \bibliographystyle{amsalpha}
\providecommand{\bysame}{\leavevmode\hbox to3em{\hrulefill}\thinspace}
\providecommand{\MR}{\relax\ifhmode\unskip\space\fi MR }
\providecommand{\MRhref}[2]{%
  \href{http://www.ams.org/mathscinet-getitem?mr=#1}{#2}
}
\providecommand{\href}[2]{#2}

\end{document}